\documentclass[a4paper,UKenglish,cleveref,autoref,numberwithinsect]{lipics-v2021}

\hypersetup{
  colorlinks=true,
  linkcolor=blue,
  citecolor=blue,
  urlcolor=blue
}

\scrollmode
\usepackage{amsmath}
\usepackage{amsfonts}
\usepackage{amssymb}
\usepackage{latexsym}
\usepackage{stmaryrd}
\usepackage{array}
\usepackage{exscale}
\newcommand{\nc}{\newcommand}
\newcommand{\ol}{\overline}
\newcommand{\es}{\emptyset}
\newcommand{\sm}{\setminus}

\newcommand{\vp}{\varphi}

\newcommand{\bc}{\bigcup}

\newcommand{\Lra}{\Leftrightarrow}

\newcommand{\Ra}{\Rightarrow}

\newcommand{\ra}{\rightarrow}

\newcommand{\lra}{\leftrightarrow}

\newcommand{\sse}{\subseteq}

\newcommand{\fa}{\forall}

\newcommand{\mr}{\mathrm}
\newcommand{\mc}{\mathcal}
\newcommand{\mf}{\mathfrak}

\newcommand{\DMO}{\DeclareMathOperator}
\newcommand{\NN}{\mathbb{N}}
\newcommand{\NNZ}{\NN_0}
\newcommand{\ZZ}{\mathbb{Z}}

\mathchardef\breakingcomma\mathcode`\,
{\catcode`,=\active
  \gdef,{\breakingcomma\discretionary{}{}{}}
}

\newcommand{\mb}{{\:|\:}} 
\newcommand{\set}[1]{\{ #1 \}}

\nc{\simlvi}[1]{\!\sim_{#1}}
\newcommand{\tb}[2]{\set{#1, \dots, #2}} 
\providecommand{\abs}[1]{\lvert #1 \rvert} 
\newcommand{\trans}[1]{#1^{\hspace{0.05em}\mr{t}}} 
\DeclareMathOperator{\length}{lgth} 
\newcommand{\Va}{\mc{V\hspace{-0.1em}A}}

\newcommand{\Lit}{\mc{LIT}}
\newcommand{\Cl}{\mc{CL}}
\newcommand{\Cls}{\mc{CLS}}

\newcommand{\Pcls}[1]{#1\mbox{--}\Cls}

\newcommand{\Sat}{\mc{SAT}}

\newcommand{\Usat}{\mc{USAT}}

\newcommand{\Musat}{\mc{M\hspace{0.8pt}U}} 
\newcommand{\Musati}[1]{\Musat_{\!#1}} 
\newcommand{\Smusat}{\mc{S}\Musat} 
\newcommand{\Smusati}[1]{\Smusat_{\!#1}}

\nc{\Clsoo}{\Cls^{1,1}} 
\DeclareMathOperator{\lit}{lit}
\DeclareMathOperator{\var}{var}

\DMO{\dos}{ds} 
\DMO{\mdos}{mds} 
\newcommand{\Clash}{\mc{HIT}} 

\newcommand{\Uclash}{\mc{U}\Clash} 
\newcommand{\Uclashi}[1]{\Uclash_{\!\!#1}}
\DMO{\premr}{ax} 
\DMO{\concr}{C} 
\DMO{\allcr}{cl} 

\DMO{\thardness}{thd} 
\DMO{\phardness}{phd} 
\DMO{\whardness}{awid} 
\DMO{\dep}{dep} 
\DMO{\hts}{hs} 
\DMO{\semspace}{css} 
\DMO{\resspace}{crs} 
\DMO{\treespace}{cts} 
\nc{\bth}[1]{\langle{#1}\rangle} 
\DMO{\rsub}{r_S} 
\DMO{\rk}{r} 
\DMO{\ro}{\rk_1} 
\DMO{\rki}{\rk_{\infty}} 
\DMO{\rpl}{r^{pl}} 
\DMO{\ropl}{\rk_1^{pl}} 
\nc{\rslur}{\xrightarrow{\text{SLUR}}} 
\nc{\rslurs}{\rslur_{\!*}} 
\DMO{\slur}{slur} 
\nc{\Slur}{\mc{SLUR}} 
\nc{\rkslur}[1]{\xrightarrow{\text{SLUR}_{#1}}} 
\nc{\rkslurs}[1]{\rkslur{#1}_{\!*}} 
\nc{\Altsluri}[1]{\Slur(#1)}
\nc{\Altslurstari}[1]{\Slur\text{\textasteriskcentered}(#1)}
\nc{\Canoni}[1]{\mr{CANON}(#1)}
\nc{\rkslurstar}[1]{\xrightarrow{\text{SLUR\textasteriskcentered}#1}} 
\nc{\rkslursstar}[1]{\rkslurstar{#1}_{\!*}} 
\DMO{\slurstar}{\slur\!\text{\textasteriskcentered}}
\nc{\Urefc}{\mc{UC}}
\nc{\Propc}{\mc{PC}}
\nc{\Wrefc}{\mc{WC}} 
\DeclareMathOperator{\mus}{MU}

\DeclareMathOperator{\ldeg}{ld} 
\DeclareMathOperator{\vdeg}{vd} 
\DeclareMathOperator{\minvdeg}{\mu\!\vdeg} 
\DMO{\varmvd}{\var_{\minvdeg}} 
\DMO{\nfc}{fc} 
\DMO{\maxnfc}{\nu\!\nfc} 
\nc{\Dt}[1]{\mc{F}_{#1}} 

\nc{\svbf}{\mc{VB}} 
\nc{\svbfs}{\mc{VB}^*} 
\DMO{\potp}{pp} 
\DMO{\potprec}{NM} 
\DMO{\minnonmer}{VDM} 
\DMO{\minnonmerh}{VDH} 
\DMO{\maxsmar}{FCM} 
\DMO{\maxsmarh}{FCH} 
\DMO{\varsing}{\var_s} 
\DMO{\varosing}{\var_{1s}} 
\DMO{\varnosing}{\var_{\neg1s}} 
\DMO{\nsv}{\mathit{n}_s} 
\DMO{\nosv}{\mathit{n}_{1s}}
\DMO{\nnosv}{\mathit{n}_{\neg1s}}
\nc{\Musatns}{\Musat'} 
\nc{\Musatnsi}[1]{\Musati{#1}'}
\nc{\Smusatns}{\Smusat'} 
\nc{\Smusatnsi}[1]{\Smusati{#1}'}
\nc{\Uclashns}{\Uclash'} 
\nc{\Uclashnsi}[1]{\Uclashi{#1}'}
\nc{\tsdp}{\xrightarrow{\text{sDP}}}
\nc{\tsdps}{\tsdp_{\!*}}
\nc{\tosdp}{\xrightarrow{\text{1sDP}}}
\nc{\tosdps}{\tosdp_{\!*}}
\DMO{\sdp}{sDP} 
\DMO{\osdp}{sDP_1} 
\nc{\cflmusat}{\mc{CF}\Musat} 
\nc{\cflmusati}[1]{\mc{CF}\Musati{#1}}
\nc{\cflimusat}{\mc{CFI}\Musat} 
\DMO{\sNF}{sNF} 
\DMO{\eqp}{eqp} 
\DMO{\sgp}{sp} 
\DMO{\singind}{si} 
\DMO{\osingind}{si_1} 
\DMO{\shyp}{svh} 
\DMO{\sdph}{ssh} 
\DMO{\msdph}{mss} 
\DMO{\osdph}{ssh_1} 
\DMO{\mosdph}{mss_1} 
\DMO{\mps}{mps} 
\DMO{\purec}{puc} 
\DMO{\doping}{D}
\nc{\glue}[4]{\operatorname{glue}((#1,#2), (#3,#4))} 
\nc{\gluea}[3]{#1 \mathbin{\boxplus}_{#3} #2} 
\DMO{\saturate}{S}
\DMO{\marginalise}{M}
\usepackage[all]{xy}
\usepackage{cancel}
\usepackage{xr}

\usepackage{algorithm}
\usepackage[commentColor=black]{algpseudocodex}

\newcommand{\Bpt}[1]{\mc{B}_{#1}}
\newcommand{\Hp}{\operatorname{U}^2}
\newcommand{\Hl}{\operatorname{U}^1}
\newcommand{\Shl}{\operatorname{U}^0}
\newcommand{\Ub}{\operatorname{U}^0}

\newcommand{\FI}{\mf{F}_I}
\newcommand{\FII}{\mf{F}_{II}}
\newcommand{\FIII}{\mf{F}_{III}}
\newcommand{\FIV}{\mf{F}_{IV}}

\newcommand{\Bmusat}{2\mbox{--}\!\Musat}

\DeclareMathOperator{\idg}{idg}
\newcommand{\Bclss}{\Pcls{2}^*}

\newcommand{\Dipcmp}{\mc{DPC}}
\newcommand{\Dicmp}{\mc{DC}}

\DeclareMathOperator{\tFC}{FC}

\DeclareMathOperator{\pfirst}{first}
\DeclareMathOperator{\plast}{last}
\newcommand{\reach}[3]{#1 \xrightarrow{#3} #2}

\DeclareMathOperator{\linord}{\mc{LO}}

\DeclareMathOperator{\enum}{enum}

\newtheorem{question}[theorem]{Question}

\title{Simple minimally unsatisfiable subsets of 2-CNFs}

\titlerunning{MUSs of 2-CNFs}

\author{Oliver Kullmann\footnote{corresponding author}}{Department of Computer Science, Swansea University, United Kingdom}{O.Kullmann@Swansea.ac.uk}{https://orcid.org/0000-0003-3021-0095}{}

\author{Edward Clewer}{Independent researcher, United Kingdom}{edward.clewer.research@gmail.com}{}{}

\authorrunning{O. Kullmann and E.Clewer}

\Copyright{Oliver Kullmann and Edward Clewer} 



\EventLogo{} 

\nolinenumbers 

\begin{document}


\maketitle

\begin{abstract}
  We present a study of minimal unsatisfiable subsets (MUSs) of 2-CNF Boolean formulas, building on the Abbasizanjani-Kullmann classification of minimally unsatisfiable 2-CNFs (2-MUs).
  First we give a linear-time procedure for recognising 2-MUs.
  Then we study the problem of finding one simple MUS.
  On the one hand we extend the results by Kleine B\"uning et al, which showed NP-completeness of the decision, whether a deficiency-1 MUS exists.
  On the other hand we show that deciding/finding an MUS containing one or two unit-clauses (which are special deficiency-1 MUSs) can be done in polynomial time.
  Finally we present an incremental polynomial time algorithm for some special type of MUSs, namely those MUSs containing at least one unit-clause.

  We conclude by discussing the main open problem, developing a deeper understanding of the landscape of easy/hard MUSs of 2-CNFs.
\end{abstract}

\section{Introduction}
\label{sec:introduction}

An unsatisfiable CNF formula in general contains many ``sources'' of unsatisfiability, and the main approach is to consider minimally unsatisfiable subsets (MUSs for short).
MUSs serve as root causes of infeasibility, and find applications in diagnosis, error localisation, product configuration, model checking (CEGAR), and redundancy removal.

Considering complete enumeration of all MUSs, an early paper is \cite{LiffitonSakallah2005AllMUS}.
An important milestone was given by \cite{LPMMS2016MUS}, which captures much of the progress in the field, and presents the MARCO framework.
A general tool set is provided in \cite{BendikCerna2020MUST}, while a probabilistic approach to approximate counting is introduced in \cite{bendik-meel-2021}.

Since we are focusing on the special input class 2-CNF, enumeration of MUSs for special classes is of importance to us.
It appears the only work available here is on Horn CNFs (at most one positive literal in each clause).
\cite[Theorem 5]{PenalozaSertkaya2010AxiomPinpointing} shows MUS enumeration with polynomial delay for Horn CNFs.
They also make a remark on a special case of Horn CNFs, which is worth discussing here.
Namely they consider ``core KBs'', which are 2-CNFs with only mixed clauses, containing one positive and one negative literal (which are always satisfiable).
On \cite[Page 282]{PenalozaSertkaya2010AxiomPinpointing} they make the remark, that core KBs allow finding minimal axiom sets with polynomial delay, by enumerating paths (see \cite{PengEtal2019PathEnumeration} for a recent paper on enumerating paths).
Our situation is more complex, since we are considering general 2-CNFs, where one has to consider the problem of complementary literals occurring on a path (``contradictory paths'').
So even in the simplest case of fixing two unit-clauses (the situation we call ``Family I''), a more elaborate algorithm for finding one MUS is needed, and enumeration with polynomial delay is not known.
And if we relax the condition, and allow more general forms of MUSs (``Family II'' in our context), then the problem of duplicated enumeration of the same MUS (by different paths) arises, so that we can only guarantee enumeration with incremental polynomial time.
\cite{ArifMenciaMarquesSilva2015MusHorn} considers Horn CNFs from the practical side.

Back to general CNF, a different framework is to search for (single) ``simple'' MUSs, where the standard interpretation is ``short(est)''.
An early paper here is \cite{OhMneimnehAndrausSakallahMarkov2004Amuse}, while a recent article presenting also a general framework is \cite{DBLP:journals/jair/GambaBG23}.

We aim at considering more intelligent complexity measures than sheer size --- an MUS might be structurally very simple, but nonetheless very long.
Based on the complexity measure ``deficiency'' for MUs (minimally unsatisfiable CNFs), we consider the simplest MUSs in this respect, which have deficiency one (that is, exactly one more clauses than variables).
Moreover we are aiming at polytime results, and thus we focus on MUSs of deficiency one of 2-CNFs.
We remark here that the complexity of enumerating arbitrary MUSs of 2-CNFs is wide open --- and the task of enumerating a special class of MUSs might be easier or harder than enumerating the general class.

So we are combining two strands of research here, concentrating on deficiency one together with 2-CNF.
A common abbreviation for MUs of deficiency one is MU(1).
\cite[Theorem 14]{DDK98} showed that whether an input CNF is in MU(1) can be decided in quadratic time (this has been extended to polytime for MU(k) in \cite{FKS00}).
A complete classification of 2-CNF-MUs, called 2-MUs, has been achieved in \cite{AbbasizanjaniKullmann20202MUa}, based on \cite{AbbasizanjaniKullmann2016MUconference}.
We only need the simplest case of the classification, the intersection of 2-MU and MU(1).
Such MUs are of four types, called ``Families I-IV'' in this paper.\footnote{In \cite{AbbasizanjaniKullmann20202MUa} the four families are only referred to by their technical notations.
  Very short after submission of this paper to the SAT2026, a much extended version of \cite{AbbasizanjaniKullmann20202MUa} will be made available on arXiv.
However what we need in this paper is completely covered in Section 4 and Appendix C of \cite{AbbasizanjaniKullmann20202MUa}.}

Before coming to our results, we mention the notion of ``digraphs with given skew-symmetry'', which can be considered as the study of 2-CNFs from a graph-theoretical angle.\footnote{This topic will be more extensively discussed in the upcoming extension of \cite{AbbasizanjaniKullmann20202MUa}.}
The fundamental paper here is \cite{GoldbergKarzanov1996SkewSymmetry}.
A ``skew-symmetry'' is like having literals as vertices, where complementation is an anti-automorphism, that is, arcs $a \ra b$ are mapped to their contraposition $\ol b \ra \ol a$.
\cite{GoldbergKarzanov1996SkewSymmetry} introduces arc-regular paths (not having some arc and its contraposition on the path), and show that the basic problems regarding such paths (including finding shortest ones) can be solved in linear time.
We need to consider vertex-regular paths (not having a literal and its complement on the path).
It is common knowledge that the question of vertex-regular paths can be reduced to it, by splitting every vertex into two (for the incoming and outgoing arcs, with one new arc connecting the two new vertices).

Finally our results:
\begin{enumerate}
\item \cref{sec:dec2mu} considers the decision problem for 2-MU.
  The basic algorithm (exploiting 2-SAT decision in linear time) takes quadratic time, and we give a linear time algorithm in \cref{thm:2MUlintime}, based on a general result on ``checked singular DP-reduction'' in \cref{thm:boundeddegree}.
\item \cref{sec:def1MUSNPcompl} presents a refined hardness analysis for MUSs of deficiency one, pinpointing more precisely the cause of complexity with the Families III and IV in Theorems \ref{thm:gendef1NPC}, \ref{thm:gendef1NPCvariant}, based on the structural results of \cref{thm:basicharactFG}.
\item \cref{Sec:def1MUSsimple} shows how to find one MUS containing some unit-clause.
\cref{thm:twounits} characterises the case with two unit-clauses, while
\cref{thm:oneunit} characterises the cases with at least one unit-clause, and \cref{cor:oneunit2} shows that finding an MUS with at least one unit-clause (i.e., belonging to Family I or II) can be done in quadratic time.
\item Finally \cref{Sec:def1MUSpolydelay} discusses the enumeration of MUS with at least one unit-clause.
For the case with given unit-clause, \cref{thm:correctnessenumwithunit} shows the correctness of the algorithm, while \cref{thm:runtimeenumwith} presents its complexity.
Similarly, \cref{thm:correctnessenumwithoutunit} shows correctness of the algorithm for handling all unit-clauses, while \cref{thm:runtimeenumwithout} presents its complexity.
\end{enumerate}

\section{Preliminaries}
\label{sec:Preliminaries}

We use $\NNZ = \set{0,1,2,\dots}$.
When stating complexity results using Big-Oh, if the term becomes zero, then this is interpreted as linear time (the minimum runtime for ordinary algorithms, in the usual random-access model).

\subsection{Clause-sets}
\label{sec:prelimcls}

Let $\Va$ be the set of variables, and $\Lit = \Va \cup \set{\ol v : v \in \Va}$ be the set of literals, where for the complementation holds $\ol{\ol x} = x$ for all $x \in \Lit$.
For $L \sse \Lit$ we use $\ol L := \set{\ol x : x \in L}$.
$L$ contains a clash iff $L \cap \ol L \ne \es$ (otherwise it is clashfree).

Now a clause $C$ is a finite $C \sse \Lit$ with $C \cap \ol C = \es$ ($C$ does not contain a clash/conflict), and the set of all clauses is $\Cl$.
Finally a clause-set is a finite subset of $\Cl$, and the set of all clause-sets is $\Cls$.
A special clause is the empty clause $\bot := \es \in \Cl$, and a special clause-set is the empty clause-set $\top := \es \in \Cls$.
For $k \in \NNZ$ we define $\Pcls k := \set{F \in \Cls \mb \fa\, C \in F : \abs C \le k}$ (all clauses have length (or ``width'') at most $k$).
Especially we use $\Bclss := \set{F \in \Pcls 2 : \bot \notin F}$.

Via $\var: \Lit \ra \Va$ we obtain the underlying variable of a literal, while for clauses $C$ we use $\var(C) := \set{\var(x) : x \in C} \sse \Va$, and $\var(F) := \bc_{C \in F} \var(C) \sse \Va$ for clause-sets $F$.
With $\lit(F) := \var(F) \cup \ol{\var(F)}$ we obtain the literal-set of $F$; so $x \in \lit(F) \Lra \set{x,\ol{x}} \cap \bc_{C \in F} C \ne \es$.
The degree of a literal $x$ in $F$ is $\ldeg_F(x) := \abs{\set{C \in F : x \in C}} \in \NNZ$, and the degree of a variable $v$ is $\vdeg_F(v) := \ldeg_F(v) + \ldeg_F(\ol v)$.
For $k \in \NNZ$ let $\var_k(F) := \set{v \in \var(F) : \vdeg_F(v) = k}$ be the set of degree-$k$-variables, and let $n_k(F) := \abs{\var_k(F)} \in \NNZ$ be their number.
Similarly we use $\lit_k(F) := \set{x \in \lit(F) : \ldeg_F(x) = k}$ and $n'_k(F) := \abs{\lit_k(F)} \in \NNZ$.
For $F \in \Cls$ a variable $v \in \var(F)$ is called \textbf{singular} if $\ldeg_F(v) = 1$ or $\ldeg_F(\ol v) = 1$, the set of all singular variables is $\varsing(F) = \set{v \in \var(F) : \set{v,\ol v} \cap \lit_1(F) \ne \es}$.
While $v$ is called \textbf{1-singular} if $\ldeg_F(v) = 1$ and $\ldeg_F(\ol v) = 1$, and the set of all 1-singular variables is $\varosing(F) = \set{v \in \var(F) : \set{v,\ol v} \sse \lit_1(F)}$.

As measures for $F \in \Cls$ we use:
\begin{itemize}
\item $n(F) := \abs{\var(F)} \in \NNZ$ (number of variables)
\item $c(F) := \abs{F} \in \NNZ$ (number of clauses)
\item $u(F) := \abs{\set{C \in F : \abs C = 1}} \in \NNZ$ (number of unit-clauses)
\item $\ell(F) := \sum_{C \in F} \abs{C} \in \NNZ$ (our measure of formula size).
\item $\delta(F) := c(F) - n(F) \in \ZZ$ (deficiency).
\end{itemize}
Concerning the trivial cases, we have $n(F) = 0 \Lra F \in \set{\top, \set{\bot}}$ and $c(F) = 0 \Lra \ell(F) = 0 \Lra F = \top$.
The only non-standard measure is $u(F)$, the number of unit-clauses.

The semantics is given by partial assignments, which are maps $\vp: V \ra \set{0,1}$, where $V \sse \Va$.
For $v \in V$ we define $\vp(\ol v) = \ol{\vp(v)}$, and thus $\vp$ is defined on $V \cup \ol V$.
A clause $C$ is satisfied by $\vp$ if there is a literal $x \in C$ with $\vp(x) = 1$, while a clause-set is satisfied by $\vp$ if all clauses $C \in F$ are satisfied by $\vp$.
The set of all satisfiable clause-sets is denoted by $\Sat \sse \Cls$, with $\top \in \Sat$, while $\Usat := \Cls \sm \Sat$ is the set of all unsatisfiable clause-sets.
A clause-set $F$ is minimally unsatisfiable (is an MU), if it is unsatisfiable, but elimination of any clause renders it satisfiable; the set of all MUs is defined as $\Musat := \set{F \in \Usat \mb \fa\, C \in F : F \sm \set{C} \in \Sat}$, with $\set{\bot} \in \Musat$.
And $\Bmusat := \Musat \cap \Pcls 2$.
We also use $\Musati{\delta=1} := \set{F \in \Musat : \delta(F) = 1}$.
Recall $\delta(F) \ge 1$ for all $F \in \Musat$; see \cite{Kullmann2007HandbuchMU2021} for more background information.
Another (trivial) fact is $\varosing(F) = \var_2(F)$ for $F \in \Musat$.

Finally an MUS of $F \in \Cls$ is an MU which is a subset of $F$, and the set of all MUSs of $F$ is $\mus(F) := \set{F' \sse F : F' \in \Musat}$.
We use $\mus_{C_1,\dots,C_m}(F) := \set{F' \in \mus(F) : C_1, \dots, C_m \in F'}$.

\subsection{Digraphs}
\label{sec:prelimgraphs}

Digraphs $G$ have finite vertex-sets $V(G)$, while the arcs are pairs of different vertices: $E(G) \sse \set{(v,w) \in V(G)^2 : v \ne w}$.
By $\trans G$ we denote the digraph obtained from $G$ by reversing the direction of the arcs.
A subdigraph $G'$ of $G$ is a digraph with $V(G') \sse V(G)$ and $E(G') \sse E(G)$.
A path $P$ in a digraph $G$ is a subdigraph with $V(P) \ne \es$, its length is denoted by $\length(P) := \abs{E(P)} \ge 0$ (the number of arcs; $\length(P) = \abs{V(P)} - 1$).
The first resp.\ last vertex in a path $P$ are denoted by $\pfirst(P), \plast(P) \in V(P)$.
The concatenation of paths $P, P'$ is written as $P;P'$, provided $\plast(P) = \pfirst(P')$ and $V(P) \cap V(P') = \set{\plast(P)}$.
As a special case, appending a vertex $v$ to a path $P$ (with $v \notin V(P)$) is denoted by $P;v$.
By $G - A$ for $A$ some set we denote the subdigraph of $G$ obtained by removing those vertices $v \in V(G)$ and their incident arcs from $G$ for which we have $v \in A$.
Finally the relation $\reach{x}{y}{G}$ is true iff there is a path in $G$ from $x$ to $y$.

\subsection{Digraphs with complementation}
\label{sec:prelimgraphscompl}

In the literature we have the notion of digraphs with given skew-symmetry, which we handle by using literals as vertices, and realising the skew-symmetry by complementation.

\begin{definition}\label{def:digraphpartcompl}
  A \textbf{digraph with partial complementation} is a digraph $G$ whose vertices are literals, i.e., $V(G) \sse \Lit$.
  The set of all digraph with partial complementation is denoted by $\Dipcmp$.
  A \textbf{digraph with complementation} is a digraph with partial complementation, such that the vertex-set is closed under complementation (i.e., $\ol{V(G)} = V(G)$), and such that the arc-set is closed under contraposition, that is, for all $(x,y) \in E(G)$ we have $\ol{(x,y)} := (\ol y, \ol x) \in E(G)$.
  The set of all digraph with complementation is denoted by $\Dicmp$.
\end{definition}

For a set $A \sse E(G)$ of arcs in $G \in \Dipcmp$, by
\begin{displaymath}
  \Cl(A) := \set{\set{\ol x, y} : (x,y) \in A} \in \Cls
\end{displaymath}
we denote the corresponding clause-set; so $\Cl(E(\idg(F))) = F$.
For a single arc $e$, by $\Cl(e)$ we denote the single element of $\Cl(\set{e})$.
The clause $\Cl(e) = \set{\ol x, y}$ for $e = (x,y)$ is a unit-clause iff $y = \ol x$, and then the arc $e = (x,y)$ is called a \textbf{unit}.

\begin{definition}\label{def:regular}
  Consider $G \in \Dipcmp$.
  A path $P$ is \textbf{regular} (more precisely ``vertex-regular'', but we do not consider arc-regular paths in this paper) if $V(P)$ is clashfree.
  We use $\mf R_{x,y}(G)$ for the set of regular paths (as subdigraphs) from $x$ to $y$.
\end{definition}

\begin{definition}\label{def:nearlyreg}
  Consider $G \in \Dipcmp$.
  A path $P$ with exactly one clash, such that $P$ without the last vertex is a regular path, is called \textbf{nearly regular}.
  For a nearly regular path $P$ there is a unique \textbf{regular decomposition} $P = P_0;P_1$, where $P_0$ is a regular path or the empty digraph, while $P_1$ is a path with $\plast(P_1) = \ol{\pfirst(P_1)}$.
  We use $\mf U_x(G)$ for the set of nearly regular paths $P$ in $\idg(F)$ starting in $x$ (``U'' stands for ``unique clash'').
\end{definition}

For a path $P$ with $\length(P) \ge 1$ in a digraph with complementation, the \textbf{contraposition $\ol P$}, obtained by contraposing all arcs and reverting their order, is a path from $\ol{\plast(P)}$ to $\ol{\pfirst(P)}$.

The fundamental example of a digraph with complementation is given by the \textbf{implication digraph} $\idg: \Bclss \ra \Dicmp$, where $\idg(F)$ has vertex-set $V(\idg(F)) = \lit(F)$ and for every clause $\set{x,y} \in F$ the arcs $(\ol x, y), (\ol y, x) \in E(\idg(F))$.
So for binary clauses we get two arcs, while for a unit-clause $\set{x}$ we just get one arc, namely $(\ol x, x)$.
Digraphs with \emph{partial complementation} are obtained from the implication digraphs by forming subdigraphs in various forms.

\section{Review of the classification of 2-MU}
\label{sec:review2MU}

For $F \in \Bmusat$ we have the following basic facts, which we will use freely in the remainder of the paper:
\begin{itemize}
\item $\lit(F) = \lit_1(F) \cup \lit_2(F)$ and $\var(F) = \var_2(F) \cup \var_3(F) \cup \var_4(F)$ (\cite{KBZ2002b}).
\item It follows $n'_1(F) + n'_2(F) = 2 n(F)$, $n_2(F) + n_3(F) + n_4(F) = n(F)$, and $n'_1(F) = 2 n_2(F) + n_3(F)$, $n'_2(F) = n_3(F) + 2 n_4(F)$.
  And also $\varsing(F) = \var_2(F) \cup \var_3(F)$.
\item $0 \le u(F) \le 2$ and $u(F) \ne 0 \Ra \delta(F) = 1$ (\cite[Lemma 8]{Liberatore2008Redundanz}).
\end{itemize}

The nonsingular elements of $\Bmusat$ for deficiency $k \ge 1$ are
\begin{description}
\item[$k=1$] $\set{\bot}$
\item[$k\ge 2$] $\Bpt k \doteq v_1 \lra v_2 \lra v_3 \lra \dots \lra v_{k-1} \lra v_k \lra \neg v_1$, with $k$ variables and $2 k$ clauses, by \cite{KBZ2002b} and \cite[Corollary 2]{AbbasizanjaniKullmann2016MUconference}.
\end{description}

\subsection{The Four Families}
\label{sec:fourfamilies}

The report \cite{AbbasizanjaniKullmann20202MUa} classifies effectively the isomorphism types of $\Bmusat$, where two clause-sets $F, F'$ are isomorphic iff $F$ can be transformed into $F'$ by a renaming of the variables (with $n(F) !$ many choices) and by a flipping of literals (with $2^{n(F)}$ many choices).
In this paper, we only need the deficiency-1-case, and only a broad classification, which we now review.

The four Families I-IV, written as classes $\FI, \FII, \FIII, \FIV$, partition the deficiency-1-subclass of $\Bmusat$.
In \cite{AbbasizanjaniKullmann20202MUa} the families are given via representatives as four families $\Hp_n, \Hl_{n,i}, \Shl_{n,i}, \Ub_{n,x,y}$, with concrete parameters, in order to pinpoint the isomorphism types, but in this paper we concentrate on the closure under isomorphism of these families, yielding the above classes.
We give now the main characteristics of the four families, together with the generic formulas (possibly with subcases) for the elements.
For the representations of the formulas we use the notation ``$x \ra \dots \ra y$'' for a chain of implications of the form $a \ra b$, where $\var(a) \ne \var(b)$ (thus the corresponding clause $\set{\ol a, b}$ has length two), and where furthermore the internal variables (other than those of the endpoints $x, y$, i.e., those variables hidden by the dots) are different from each other, and also from those of other internal variables in such chains (in the same formula).
The literals $x, y \in \Lit$ are arbitrary, but fulfil $\var(x) \ne \var(y)$, and $F \in \Bmusat$ (we do not need to restrict to $\delta(F) = 1$, but this is always implied by the conditions):
\begin{enumerate}
\item[I] $F \in \FI \Lra n_3(F) = 0 \land n_4(F) = 0 \Lra u(F) = 2$.
  \begin{enumerate}
  \item[Ia] $\boxed{\set{x}, \set{\ol x}}$
  \item[Ib] $\boxed{\set{x}, \set{y}}$ + $\boxed{x \ra \dots \ra \ol y}$
  \end{enumerate}
\item[II] $F \in \FII \Lra n_3(F) = 1 \Lra n_3(F) = 1 \land n_4(F) = 0 \Lra u(F) = 1$.
  \begin{enumerate}
  \item[IIa] $\boxed{\set{x}}$ + $\boxed{x \ra \dots \ra \ol x}$ ($\var_3(F) = \set{\var(x)}$)
  \item[IIb] $\boxed{\set{x}}$ + $\boxed{x \ra \dots \ra y \ra \dots \ra \ol y}$ ($\var_3(F) = \set{\var(y)}$)
  \end{enumerate}
\item[III] $F \in \FIII \Lra n_3(F) = 0 \land n_4(F) = 1$.\\
  $\boxed{x \ra \dots \ra \ol x \ra \dots \ra x}$ ($\var_4(F) = \set{\var(x)}$)
\item[IV] $F \in \FIV \Lra n_3(F) = 2 \land n_4(F) = 0$.\\
  $\boxed{x \ra \dots \ra \ol x \ra \dots \ra y \ra \dots \ra \ol y}$ ($\var_3(F) = \set{\var(x), \var(y)}$)
\end{enumerate}

\section{Deciding 2-MU}
\label{sec:dec2mu}

Our first result is that 2-MU can be decided in linear time (while the trivial algorithm, using SAT-decision for 2-CNF in linear time, needs quadratic time).
Our main tool is \emph{checked singular DP-reduction}, and the algorithm for deciding 2-MU runs checked singular DP-reduction, and if the reduction was successful, checks the final result (depending on the deficiency $1$ or $\ge 2$).
Singular DP-reduction has been applied in the literature to reduce MUs, see for example \cite{KullmannZhao2012ConfluenceJ} for a fundamental study.
\cite[Lemma 9]{KullmannZhao2012ConfluenceJ} indeed spells out the conditions such that we can check for MU, and this is our basis for ``checked'' sDP-reduction (which may fail):
\begin{definition}\label{def:checkedsDP}
  Recall that for $F \in \Cls$ a variable $v \in \var(F)$ is called singular if $\ldeg_F(v) = 1$ or $\ldeg_F(\ol v) = 1$.
  A \textbf{main clause} is the clause $C \in F$ containing the chosen literal $v$ resp.\ $\ol v$ occurring only once (there is a choice iff $\ldeg_F(v) = \ldeg_F(\ol v) = 1$), while the \textbf{side clauses} $D_1, \dots, D_m \in F$ for $m \ge 0$ are the clauses (without duplication) containing the complemented literal $\ol v$ resp.\ $v$.
  The \textbf{sDP-check} for $v, F$ \textbf{fails} if one of the following conditions holds:
  \begin{enumerate}
  \item[(i)] $m = 0$;
  \item[(ii)] there is $i \in \tb 1m$ such that $C$ and $D_i$ clash in a variable other than $v$;
  \item[(iii)] there are $i, j \in \tb 1m$, $i \ne j$, such that $D_i \sm C = D_j \sm C$;
  \item[(iv)] there is $i \in \tb 1m$ and $E \in F$ with $v \notin \var(E)$ and $D_i \sm C = E \sm C$.\footnote{This condition is slightly stronger than Condition 3(c) in \cite[Lemma 9]{KullmannZhao2012ConfluenceJ}:
      There we check whether the resolvent of $C$ with $D_i$ equals an existing clause $E$, while here we check whether an existing clause $E$ subsumes the resolvent.
      This is done here to maintain a consistent style for the conditions.}
  \end{enumerate}
  For $F \in \Cls$ and $v \in \Va$ a \textbf{checked sDP-reduction step} (csDP-reduction) results in FAIL, if the sDP-check fails, and otherwise results in
  \begin{displaymath}
    F' := (F \sm \set{C,D_1,\dots,D_m}) \cup \set{(C \cup D_i) \sm \set{v, \ol v} : i \in \tb 1m} \in \Cls.
  \end{displaymath}
  A \textbf{full csDP-reduction} for $F \in \Cls$ applies csDP-reduction (non-deterministically) as long as either a failure is reached, in which case FAIL is the result, or otherwise until no singular variable exists anymore, in which case the result is the final $F' \in \Cls$.
\end{definition}
The following fundamental property of csDP-reduction follows directly from \cite[Lemma 9]{KullmannZhao2012ConfluenceJ} (and is also easy to check directly):
\begin{lemma}\label{lem:csDPprop}
  For $F \in \Cls$ holds:
  \begin{enumerate}
  \item If full csDP-reduction fails, then $F \notin \Musat$.
  \item Otherwise for the resulting $F'$ holds $F \in \Musat \Lra F' \in \Musat$.
  \end{enumerate}
\end{lemma}

\begin{question}\label{que:complsDP}
  Computing full checked sDP-reduction on arbitrary clause-sets $F \in \Cls$, or full (unchecked) sDP-reduction on MUs, has an upper time-bound of $O(n(F)^2 \cdot c(F))$, i.e., cubic time:
  Each step reduces the number of variables and the number of clauses by $1$, and updating the literal degrees can be done in linear time, however the growth of the clause-lengths can add $O(n \cdot c)$ many literal occurrences in each step.
  Can this upper bound be improved, or can a lower bound via SETH be shown?
\end{question}

With 2-MU, we have a very special case, where every literal occurs at most twice (and this is maintained by checked sDP-reduction).
We consider the somewhat more general case of (maintained) bounded literal-degrees, and show that here (full) checked sDP-reduction can be performed in linear time:
\begin{theorem}\label{thm:boundeddegree}
  Consider $k \in \NN$ and $\mc C \sse \Pcls k$, such that all literal degrees for $F \in \mc C$ are also bounded by $k$, and $\mc C$ is stable under checked sDP-reduction.
  Then full checked sDP-reduction for $F \in \mc C$ can be performed in linear time.
\end{theorem}
\begin{proof}
Using standard data structures, we can initially compute the array, which yields for each literal its occurrences in clauses, and the array of the literal-degrees, all in linear time.
Finding all (initial) singular variable can thus be done also in linear time.
Now for each singular variable, performing the checks (using the literal-clause occurrence lists; note that the most costly step, checking Condition (iv) from Definition \ref{def:checkedsDP}, takes time $O(k^3)$), and in the positive case replacing the main and side clauses with their resolvents, and updating the literal degrees (together with determining and finding a new singular variable), can all be done in constant time (in $k$).
\end{proof}

\begin{theorem}\label{thm:2MUlintime}
  Whether for $F \in \Pcls 2$ holds ``$F \in \Musat$ ?'' can be decided in linear time.
\end{theorem}
\begin{proof}
By \cref{sec:review2MU} we can apply \cref{thm:boundeddegree} (with $k=2$), and compute the full checked sDP-reduction for $F$ in linear time, obtaining $F'$.
If the reduction fails or $k := \delta(F) \le 0$, then $F \notin \Musat$.
Otherwise $F \in \Musat$ iff $F' = \set{\bot}$ for $k = 1$, or else $F'$ is isomorphic to $\Bpt k$.
These checks can also be performed in linear time.
\end{proof}

\begin{question}\label{que:2MUlintime}
  Finding an element of $\mus(F)$ for $F \in \Pcls 2$ can be done in quadratic time, by the trivial algorithm:
  \begin{algorithmic}
    \Require $F \in \Usat$
    \Ensure $F' \in \mus(F)$
    \Function{find}{$F \in \Pcls 2$} \Comment{works indeed for all $F \in \Cls$}
      \State {$F' \gets F$}
      \For{$C \in F$}
        \If{$F' \sm \set{C} \in \Usat$}
          $F' \gets F' \sm \set{C}$
        \EndIf
      \EndFor
      \State \Return{$F'$}
    \EndFunction
  \end{algorithmic}
  Is there an asymptotically faster algorithm?
  Given that we can decide 2-MU in linear time (Theorem \ref{thm:2MUlintime}), perhaps we can even find a 2-MU in linear time?
\end{question}

\section{NP-completeness of finding simple MUSs}
\label{sec:def1MUSNPcompl}

We now start investigating ``simple'' MUSs of 2-CNFs $F$.
Considering the complexity parameter \emph{deficiency} for MU, the simplest MUSs of $F$ have deficiency one.
\cite{KBWS20162CNFJ} showed that deciding the existence of an MUS of deficiency one, that is, whether $\mus(F) \cap \Musati{\delta=1} \ne \es$ holds, is NP-complete, based on the NP-completeness of deciding whether a digraph contains two vertex-disjoint paths between two given pairs of vertices.
We can give now a direct proof of this (not using ``read-once resolution'' as in \cite{KBWS20162CNFJ}), directly using that Families III, IV (recall Subsection \ref{sec:fourfamilies}) contain two ``independent'' paths.
This allows also to pinpoint the source of hardness more precisely, and we will later see, that Families I, II, containing only one ``independent'' path, can be found efficiently.

We start by considering special digraphs, with fixed vertices $s, t$ (and for convenience use variables as vertices, excluding the special variable $x_0$ used to construct the Family-III-MUS):
\begin{definition}\label{def:stdigraphs}
  Consider three distinct variables $s, t, x_0 \in \Va$ (fixed from now on).
  An \textbf{st-digraph} is a digraph $G$ with $V(G) \sse \Va$, $x_0 \notin V(G)$, and $s, t \in V(G)$, such that $(s,t) \notin E(G)$ and $(t,s) \notin E(G)$.
\end{definition}

Formally an st-digraph is a digraph with partial complementation, but we only consider positive literals (i.e., variables) here, and st-digraphs don't have negative literals.

\begin{definition}[{\cite[Theorem 4.1]{KBWS20162CNFJ}}]\label{def:tFC}
  Consider an st-digraph $G$.
  The clause-set $\tFC(G) \in \Pcls 2$ (without unit-clauses or the empty clause) is defined as follows:
  \begin{enumerate}
  \item Let $t: V(G) \ra (V(G) \sm \set{s,t}) \cup \set{x_0, \ol{x_0}}$ be defined by
    \begin{displaymath}
      t(v) :=
      \begin{cases}
        x_0 & \text{if } v = s\\
        \ol{x_0} & \text{if } v = t\\
        v & \text{otherwise}
      \end{cases}.
    \end{displaymath}
  \item We use furthermore $t((a,b))$ for $a, b \in V(G)$ for the clause $\set{\ol{t(a)}, t(b)}$.
  \item Now $\tFC(G) := \set{t(e) : e \in E(G)}$.
  \end{enumerate}
\end{definition}

We have $\var(\tFC(G)) \sse (V(G) \sm \set{s,t}) \cup \set{x_0}$.
The arc $(s,t)$ would yield the unit-clause $\set{\ol{x_0}}$, while the arc $(t,s)$ would yield $\set{x_0}$ --- we excluded these two arcs in order to avoid having to handle these trivial cases.
Expanding on \cite[Definition 4.2]{KBWS20162CNFJ}:
\begin{definition}\label{def:CDPP}
  A \textbf{C-DPP instance} is an st-digraph $G$:
  \begin{itemize}
  \item $G$ \textbf{has a special closed walk} if there is a path $P_1$ from $s$ to $t$ and a path $P_2$ from $t$ to $s$.
  \item If such paths exist with $V(P_1) \cap V(P_2) = \set{s,t}$, then we say that $G$ \textbf{has a special cycle}.
  \end{itemize}
\end{definition}
Here ``DPP'' stands for ``disjoint path problem'', while ``C'' stands for ``cyclic''.
For examples see \cref{sec:appendixexamplescycles}.

We are ready to prove that $\tFC(G)$ encodes the two special paths from a special cycle into the two implication chains constituting the MUs of Family III (recall Subsection \ref{sec:fourfamilies}):
\begin{theorem}\label{thm:basicharactFG}
  Consider a C-DPP instance $G$:
  \begin{enumerate}
  \item\label{thm:basicharactFG1} Without both vertices $s, t$ having incoming as well as outgoing arcs, $\tFC(G)$ is satisfiable.
  \item\label{thm:basicharactFG2} Every $F' \in \mus(\tFC(G))$ contains four occurrences of variable $x_0$ (i.e., $\vdeg_{F'}(x_0) = 4$).
  \item\label{thm:basicharactFG3} $G$ has a special closed walk iff $\tFC(G)$ is unsatisfiable.
  \item\label{thm:basicharactFG4} $G$ has a special cycle iff $\mus(\tFC(G)) \cap \FIII \ne \es$.
  \item\label{thm:basicharactFG5} If $G$ does not have a special cycle, then for each $F' \in \mus(\tFC(G))$ we have $\delta(F') \ge 2$.
  \end{enumerate}
\end{theorem}
\begin{proof}
For Part \ref{thm:basicharactFG1}, the satisfying assignments for the four cases are as follows, using $V := V(G) \sm \set{s,t}$, and mapping all elements of $V$ to the same value as indicated, and stating the missing clauses of $\tFC(G)$ which would be violated by the assignment for any $v \in V$:
\begin{itemize}
\item no outgoing arc from $s$: $x_0 \mapsto 1$, $V \mapsto 0$; violates $x_0 \ra v \doteq \set{\ol{x_0},v}$;
\item no incoming arc to $s$: $x_0 \mapsto 0$, $V \mapsto 1$; violates $v \ra x_0 \doteq \set{x_0, \ol v}$;
\item no outgoing arc from $t$: $x_0 \mapsto 0$, $V \mapsto 0$; violates $\ol{x_0} \ra v \doteq \set{x_0, v}$;
\item no incoming arc to $t$: $x_0 \mapsto 1$, $V \mapsto 1$; violates $v \ra \ol{x_0}\doteq \set{\ol{x_0}, \ol v}$.
\end{itemize}
Part \ref{thm:basicharactFG2} follows from Part \ref{thm:basicharactFG1} and the fact, that all four (violated) clause-types must be present.
For Part \ref{thm:basicharactFG3} recall that any $F \in \Bclss$ is unsatisfiable iff $\idg(F)$ has a contradictory walk (a walk in the implication digraph from some literal $x$ to $\ol x$ and back).
If there is a special closed walk in $G$, then this yields immediately a contradictory walk in $\idg(\tFC(G))$ from $x_0$ to $\ol{x_0}$ and back, and thus $\tFC(G)$ is unsatisfiable.
For the other direction, by Part \ref{thm:basicharactFG1}, every contradictory walk in $\idg(\tFC(G))$ must contain $x_0$ as well as $\ol{x_0}$, and thus contains a contradictory walk $P_1$ from $x_0$ to $\ol{x_0}$ and back via $P_2$.
If $P_1$ or $P_2$ contains some complemented literal $\ol v$ for $v \in V$, then all internal vertices must be of this type, and then the contraposition $\ol{P_1}$ resp.\ $\ol{P_2}$ consists only of internal vertices from $V$, and we obtain a special closed walk as required.

Part \ref{thm:basicharactFG4} follows immediately from the characterisation of Family-III-MUSs (Subsection \ref{sec:fourfamilies}) as implication chains $x \ra \dots \ra \ol x \ra \dots \ra x$, where the literals in the dotted parts all have underlying new variables (using again the above argument, that we can always make $P_1, P_2$ to contain only positive literals).
Finally, for Part \ref{thm:basicharactFG5} assume that $G$ does not have a special cycle, and consider $F' \in \mus(\tFC(G))$.
If $\delta(F') < 2$, then indeed $\delta(F') = 1$, and so $F'$ must be one of the Families I-IV.
Now Families I, II are excluded since $F'$ does not have unit-clauses, while Family IV is excluded since $F'$ must contain a degree-$4$ variable (Part \ref{thm:basicharactFG2}).
And finally Family III is excluded by Part \ref{thm:basicharactFG4}, and thus $\delta(F') \ge 2$ must hold.
\end{proof}

We obtain the following generalisation of \cite[Corollary 4.3]{KBWS20162CNFJ} (deciding, whether a 2-CNF contains a deficiency-1-MUS, is NP-complete):
\begin{theorem}\label{thm:gendef1NPC}
  Consider any class $\FIII \sse \mf C \sse \Musati{\delta=1}$.
  Deciding whether $F \in \Pcls 2$ has an MUS belonging to $\mf C$, that is, whether $\mus(F) \cap \mf C \ne \es$ holds, is NP-complete.
\end{theorem}
\begin{proof}
Given an C-DPP instance $G$ as input, we transform it to $\tFC(G)$ in polynomial time (indeed in linear time).
Now $G$ has a special cycle iff $\tFC(G) \in \mf C$, by Theorem \ref{thm:basicharactFG}, Parts \ref{thm:basicharactFG4} and \ref{thm:basicharactFG5}.
The assertion follows now by the NP-completeness of the decision-problem, whether a C-DPP instance has a special cycle (\cite[Lemma 4.1]{KBWS20162CNFJ}).
\end{proof}

By small modifications, we can show that Family IV is similarly hard to find:
\begin{theorem}\label{thm:gendef1NPCvariant}
  Consider any class $\FIV \sse \mf C \sse \Musati{\delta=1}$.
  Deciding whether $F \in \Pcls 2$ has an MUS belonging to $\mf C$, that is, whether $\mus(F) \cap \mf C \ne \es$ holds, is NP-complete.
\end{theorem}
\begin{proof}
We introduce a new variable $y_0$, and use the modified translation $\tFC'(G)$, which adds the new binary clause $\ol{x_0} \ra y_0$, and for outgoing arcs $(t,v)$ uses the clause $y_0 \ra v$ (instead of $\ol x_0 \ra v$), and for incoming arcs $(v,s)$ uses the clause $v \ra \ol{y_0}$ (instead of $v \ra x_0$).

While with Theorem \ref{thm:basicharactFG} we obtain that every MUS of $\tFC(G)$ contains four clauses of the types $x_0 \ra v$, $v \ra x_0$, $v \ra \ol{x_0}$, $\ol{x_0} \ra v$, now every MUS of $\tFC'(G)$ contains $\ol{x_0} \ra y_0$ (yielding also the contraposed arc $\ol{y_0} \ra x_0$) plus the four clauses of the types $x_0 \ra v$, $v \ra \ol{y_0}$, $v \ra \ol{x_0}$, $\ol{y_0} \ra v$.
Thus $x_0$ occurs positively once and negatively twice, while $y_0$ occurs positively twice and negatively once, that is, we have two degree-3-variables.
With the similar proofs as for Theorem \ref{thm:basicharactFG} we obtain now, that $G$ has a special cycle iff $\mus(\tFC'(G)) \cap \FIV \ne \es$, and if $G$ does not have a special cycle, then for each $F' \in \mus(\tFC'(G))$ we have $\delta(F') \ge 2$.
So we can adapt the proof of Theorem \ref{thm:gendef1NPC}.
\end{proof}

For an example of the translation $\tFC'(G)$ see \cref{sec:appendixexamplescycles}.

\section{Finding simplest MUSs}
\label{Sec:def1MUSsimple}

In the previous section we have seen that deciding the existence of some MUS belonging to Family III or IV of a 2-CNF is NP-complete.
Now we turn to Family I (two unit-clauses; recall 2-MU has at most two unit-clauses) and Family II (one unit-clause), and show how to find them efficiently.

\subsection{Two unit-clauses}
\label{sec:twounits}

Via the regular paths from $x$ to $\ol y$ in the implication digraph (denoted by $\mf R_{x, \ol y}$; recall \cref{def:regular}) we get an efficient handle on MUSs containing the unit-clauses $\set{x}, \set{y}$:
\begin{theorem}\label{thm:twounits}
  Consider $F \in \Bclss$ and $\set{x}, \set{y} \in F$, $\set{x} \ne \set{y}$.
  Then the map $P \in \mf R_{x, \ol y}(\idg(F)) \mapsto F(P) := \Cl(E(P)) \cup \set{\set{x},\set{y}} \sse F$ is a bijection from $\mf R_{x, \ol y}(\idg(F))$ to $\mus_{\set{x},\set{y}}(F)$.
\end{theorem}
\begin{proof}
The characterisation of Family I in Subsection \ref{sec:fourfamilies} shows that Family I handles exactly the case of two unit-clauses, and that the regular paths $P$ from $x$ to $\ol y$ (note that $P$ does not contain a unit-arc, due to regularity) via the map $P \mapsto F(P)$ cover exactly all elements of $\mus_{\set{x},\set{y}}(F)$.
It remains to show that we don't have double coverage, so consider two paths $P, P' \in \mf R_{x, \ol y}(\idg(F))$ with $F(P) = F(P')$, but $P \ne P'$.
Starting from the first arc in $P$, we can show by induction that $P' = \ol P$ follows, but $\pfirst(\ol P) = \ol{\plast(P)} = y \ne x$, a contradiction.
\end{proof}

\begin{corollary}\label{cor:twounits}
  $F \in \Bclss$ has a MUS containing two different unit-clauses $\set{x}, \set{y} \in F$ iff there is a regular path in the implication digraph of $F$ from $x$ to $\ol y$.
  This condition can be checked in linear time by the algorithm from \cite{GoldbergKarzanov1996SkewSymmetry} (and the MUS extracted in the positive case).
  Also finding a shortest such MUS can be done in linear time, where ``shortest'' interpreted as minimising $n(F'), c(F'), \ell(F')$ are all equivalent.
\end{corollary}

By running through all pairs of unit-clauses we obtain:
\begin{corollary}\label{cor:twounits2}
  Whether $F \in \Bclss$ has an MUS containing exactly two unit-clauses (that is, an MUS belonging to Family I) can be decided in time $O(u(F)^2 \ell(F))$, and thus in cubic time in the input (and the MUS extracted in the positive case).
  Also finding a shortest such MUS can be done with the same time bounds.
\end{corollary}

\subsection{One unit-clause}
\label{sec:oneunit}

Via the nearly regular paths starting with $x$ (denoted by $\mf U_x$; recall \cref{def:nearlyreg}), we get an efficient handle on MUSs containing $\set{x}$ (though we need to take duplication into account).

We use here the following simple properties of paths $P$ from $x$ to $\ol x$ (always in a given implication digraph):
First we note $\length(P) \ge 1$, with $\length(P) = 1$ iff $P$ just consists of a unit, corresponding to the unit-clause $\ol x$.
The important trivial observation is now that the contraposition $\ol P$ is also a path from $x$ to $\ol x$, where $P = \ol P$ iff $\length(P) = 1$.

\begin{theorem}\label{thm:oneunit}
  Consider $F \in \Bclss$ and $\set{x} \in F$.
  Define the map $P \in \mf U_x(\idg(F)) \mapsto F(P) := \Cl(E(P)) \cup \set{\set{x}} \sse F$.
  \begin{enumerate}
  \item\label{thm:oneunit1} The map $P \mapsto F(P)$ is a surjection from $\mf U_x$ to $\mus_{\set{x}}(F)$.
  \item\label{thm:oneuni2} Recall that the preimage $F^{-1}(F')$ is the set of all $P \in \mf U_x$ with $F(P) = F'$; we will just speak here of ``the preimage''.
  \item\label{thm:oneunit3} The preimage of $F' \in \mus_{\set{x}}(F)$ is unique iff $F'$ contains exactly two unit-clauses (i.e., belongs to Family I).
    Those $F'$ are exactly the $F(P)$ for those $P \in \mf U_x$, where the final arc in $P$ is unit (and thus corresponds to the second unit-clause $\set{y}$; if we remove that final arc, then we are exactly in the situation of \cref{thm:twounits}, obtaining an element of $R_{x, \ol y}$).
  \item\label{thm:oneunit4} For all $F' \in \mus_{\set{x}}(F)$ containing exactly one unit-clause (namely $\set{x} \in F'$), there are exactly two preimages:
    \begin{enumerate}
    \item\label{thm:oneunit4a} Consider some $P \in \mf U_x$ with $F(P) = F'$.
    \item\label{thm:oneunit4b} If $\plast(P) = \ol x$ (note $\length(P) \ge 2$), then $F'$ belongs to Family IIa, and we have $\ol P \ne P$ with $\ol P \in \mf U_x$ and $F(P) = F(\ol P)$, while the preimage of $F'$ is $\set{P,\ol P}$.
    \item\label{thm:oneunit4c} Otherwise $F'$ belongs to Family IIb.
      Consider the regular decomposition $P_0;P_1 = P$ (\cref{def:nearlyreg}).
      Let $P' := P_0;\ol{P_1}$.
      We have $P' \ne P$ and the preimage of $F'$ is $\set{P, P'}$.
    \end{enumerate}
  \end{enumerate}
\end{theorem}
\begin{proof}
The characterisations of Families I, II in Subsection \ref{sec:fourfamilies} show that the MUSs containing $\set{x}$ are covered exactly by the nearly regular paths $P \in \mf U_x(\idg(F))$ via the mapping $P \mapsto F(P)$; this shows Part \ref{thm:oneunit1}.
The next assertion, Part \ref{thm:oneunit3}, follows from Theorem \ref{thm:twounits} (and the assertion in Part \ref{thm:oneunit4}, that there we have always at least two different preimages).
It remains to show Part \ref{thm:oneunit4}.
For Case \ref{thm:oneunit4b} by the same argument as in Theorem \ref{thm:twounits} it follows that the only other preimage is $\ol P$.
Finally, for Case \ref{thm:oneunit4c}, the first arc of $P_0$ is unique, since $\ol x$ can't be on the path (it would constitute Family IIa, which is of a different isomorphism type).
Then it follows by induction that the other arcs of $P_0$ are unique.
And that the only alternative to $P_1$ is $\ol{P_1}$ follows again as in Theorem \ref{thm:twounits}.
\end{proof}

\begin{corollary}\label{cor:oneunit}
  $F \in \Bclss$ has a MUS containing the unit-clause $\set{x} \in F$ iff there is a path in the implication digraph of $F$ from $x$ to $\ol x$.
  This condition can be checked in linear time (and the MUS extracted in the positive case).
\end{corollary}
\begin{proof}
If there is $F' \in \mus_{\set{x}}(F)$, then there is a path $P \in \mf U_x$ with $F(P) = F'$.
If $\plast(P) \ne \ol x$, then $P$ is the concatenation of a path $P_0$ from $x$ to $y$ and of a path $P_1$ from $y$ to $\ol y$: the concatenation $P;\ol{P_0}$ is a path from $x$ to $\ol x$.
In the other direction, if there is a path from $x$ to $\ol x$, taking the initial subpath until the first clash yields an element of $\mf U_x$.
\end{proof}

Just finding (exactly) an element of Family IIa is achieved by searching for a regular path from $x$ to $y$ (thus $\var(y) \ne \var(x)$) for such $y$ with $(y,\ol x) \in E(\idg(F))$ (that is, $\set{\ol y, \ol x} \in F$).

\begin{question}\label{que:oneunit}
  Can we efficiently find a shortest MUS using $\set{x}$?
  Different from Corollary \ref{cor:twounits}, that seems difficult here, due to the open-endedness of the set of paths $\mf U_x$.
\end{question}

If in Corollary \ref{cor:oneunit} we want to test for an MUS containing exactly one unit-clause (namely $\set{x}$), that is, for an element of Family II, then we can achieve this by removing all other unit-clauses from $F$ as preprocessing.
On the other hand, by running through all unit-clauses we obtain:
\begin{corollary}\label{cor:oneunit2}
  Whether $F \in \Bclss$ has an MUS containing at least one unit-clause (that is, an MUS belonging to Family I or II) can be decided in time $O(u(F) \ell(F))$, and thus in quadratic time in the input (and the MUS extracted in the positive case).
  The same holds for deciding (and finding) for an MUS with exactly one unit-clause (Family II).
\end{corollary}

\section{Enumeration of MUSs with a unit-clause}
\label{Sec:def1MUSpolydelay}

Based on \cref{thm:oneunit}, we can efficiently enumerate the MUSs using at least one unit-clause.
The main procedure is
\begin{displaymath}
  \enum(F \in \Bclss, \set{x} \in F, L \in \linord(\lit(F))),
\end{displaymath}
where for a set $X$ by $\linord(X)$ we denote the set of linear orders of $X$.
The order $L \in \linord(\lit(F))$ on the literals of $F$ creates first a linear order on the paths of $\idg(F)$, and then a linear order on $\bc_{\set{x} \in F} \mus_{\set{x}}(F)$, the set of all MUSs of $F$ using at least one unit-clause, as follows:

\begin{definition}\label{def:pathlex}
  Consider $F \in \Bclss$ and $L \in \linord(\lit(F))$.
  \begin{enumerate}
  \item For paths $P, P'$ we use lexicographic order based on $L$, that is, $P <_L P'$ if, considering paths as words over $\lit(F)$, for the first index $i$ from the left where $P$ differs from $P'$, we have that $P_i$ is before $P_i'$ w.r.t.\ $L$.
    This yields a strict linear order on all paths.
  \item\label{def:pathlex2} For $F \in \bc_{\set{x} \in F} \mus_{\set{x}}(F)$ there are at most two (different) literals $x, x'$ with $F \in \mus_{\set{x}}(F) \cap \mus_{\set{x'}}(F)$, and then we let $x(F)$ be the smaller one w.r.t.\ $L$.
  \item Furthermore by \cref{thm:oneunit} there are at most two paths $P, P' \in \mf U_{x(F)}$ in the preimage, and then we let $P(F)$ be the smaller one w.r.t.\ $<_L$.
  \item Finally for $F, F' \in \bc_{\set{x} \in F} \mus_{\set{x}}(F)$ we let $F <_L F'$ iff $P(F) <_L P(F')$.
  \end{enumerate}
  We call the order $<_L$ on paths resp.\ MUSs with a unit-clause the \textbf{$L$-pathlex order}.
\end{definition}

\noindent
Now the $\enum(F,x,L)$-procedure enumerates $\mus_{\set{x}}(F)$ in $L$-pathlex-order, by a recursively repeated (and guarded) DFS-transversal of $\idg(F)$ (inside DFS the output actually happens):
\begin{algorithm}[H]
  \caption{Enumeration of $\mus_{\set{x}}(F)$ in $L$-pathlex order}
  \label{alg:enum}
  \begin{algorithmic}[1]
    \Procedure{enum}{$F \in \Bclss, \set{x} \in F, L \in \linord(\lit(F)$}
      \Comment{global variables: $x, L$}
      \State $G \gets \idg(F)$ \Comment{global variable: $G$}
      \If{$\reach{x}{\ol x}{G}$}
        \label{line:initialcheck}
        \State $\Call{DFS}{(\es,\es), x}$
        \label{line:initialDFS}
      \EndIf
    \EndProcedure
  \end{algorithmic}
\end{algorithm}

Note that in order to enumerate only those MUSs using no other unit-clause than $\set{x}$, one just needs to remove all other unit-clauses from $F$.
Now the recursive procedure, where the actual work is done, is as follows (recall $\mf U_x$ from \cref{thm:oneunit}):
\begin{algorithm}[H]
  \caption{DFS-procedure for enumerating $P \in \mf U_x(G)$ in $L$-pathlex order, delivered to \texttt{OUTPUT}}
  \label{alg:DFS}
  \begin{algorithmic}[1]
    \Procedure{DFS}{$P,y$}
      \State $P \gets P;y$
      \label{line:extP}
      \State $V \gets V(P)$ \Comment{just an abbreviation}
      \If{$V$ contains a clash}
        \State $\Call{output}{P}$
        \label{line:calloutput}
      \Else
        \State $R \gets$ set of out-neighbours $z$ of $y$ in $G - V$ with $\reach{z}{\ol x}{G - V}$.
        \label{line:computeR}
        \LComment{The computation of $R$ can happen by a single run of BFS with start vertex $\ol x$ in the reversed digraph $\trans{(G - V(P))}$.}
        \For{each vertex $z \in R$, in order $L$}
          \label{line:zinR}
          \State $\Call{DFS}{z,P}$
          \label{line:recursiveDFS}
        \EndFor
      \EndIf
    \EndProcedure
  \end{algorithmic}
\end{algorithm}

Directly by code-inspection we can prove:
\begin{lemma}\label{lem:invariantDFS}
  All calls to procedure \textsc{DFS}$(P,y)$ (given in Algorithm \ref{alg:DFS}), namely the initial call in Algorithm \ref{alg:enum}, Line \ref{line:initialDFS}, and the recursive call in Algorithm \ref{alg:DFS}, Line \ref{line:recursiveDFS}, fulfil the following invariants on the two arguments $P, y$:
  \begin{enumerate}
  \item $P$ is a clashfree subdigraph of $G$, which is either empty or a path.
  \item $y \in V(G) \sm V(P)$, such that:
    \begin{itemize}
    \item if $P$ is not empty: $(\plast(P), y) \in E(G)$;
    \item $\reach{y}{\ol x}{G - V(P)}$.
    \end{itemize}
  \end{enumerate}
  Thus for $P' := P;y$ (corresponding to Line \algref{alg:DFS}{line:extP}) there exists $P'' \in \mf U_x(G)$ extending $P'$ (possibly $P'' = P'$).
\end{lemma}

Finally the actual output:
\begin{algorithm}[H]
  \caption{Output-procedure for a path $P \in \mf U_x$}
  \label{alg:output}
  \begin{algorithmic}[1]
    \Procedure{output}{$P$}
      \If{the last arc of $P$ is a unit}
        \State $\Call{print}{P}$ \Comment{Family I}
      \Else
        \State Let $P = P_0;P_1$ be the regular decomposition of $P$.
        \State $P' \gets P_0;\ol{P_1}$
        \If{$P <_L P'$}
          \label{line:ifplspp}
          \State $\Call{print}{P}$ \Comment{Family II}
        \EndIf
      \EndIf
      \EndProcedure
  \end{algorithmic}
  \begin{algorithmic}[1]
    \Procedure{print}{$P$}
      \State print CNF-header: $n = \length(P)$, $c = \length(P) + 1$
      \State print $\set{x}$
      \For{$(y,z) \in E(P)$}
        \If{$(y,z)$ is unit}
          \State print $\set{z}$ \Comment{only for Family I}
        \Else
          \State print $\set{\ol y,z}$
        \EndIf
      \EndFor
    \EndProcedure
  \end{algorithmic}
\end{algorithm}

\begin{theorem}\label{thm:correctnessenumwithunit}
  \cref{alg:enum}, implementing $\enum(F, \set{x}, L)$, indeed enumerates the set $\mus_{\set{x}}(F)$ in $L$-pathlex order.
\end{theorem}
\begin{proof}
First we need to show that in the DFS-Algorithm \ref{alg:DFS}, for the call \textsc{output}$(P)$ (Line \ref{line:calloutput}) the paths $P$ are exactly the elements of $\mf U_x$, and they appear in $L$-pathlex order.
That $P$ is always an element of $\mf U_x$ follows with Lemma \ref{lem:invariantDFS}.
And that we do not miss any path follows easily by induction: starting with $x$ (Line \algref{alg:enum}{line:initialDFS}), and then always considering all out-neighbours $z$ (Line \algref{alg:DFS}{line:computeR}), only using the necessary condition of being able to reach $\ol x$ (without using already used vertices) to cut off branches.
Finally we follow the $L$-pathlex order due to running through all out-neighbours following the order $L$ (Line \algref{alg:DFS}{line:zinR}).

It remains to show that the output of the MUSs corresponding to the paths $P$, by \textsc{output} and \textsc{print} in Algorithm \ref{alg:output}, are done correctly, and this follows immediately by the detailed statements of \cref{thm:oneunit}.
\end{proof}

For the enumeration of the underlying paths we have polynomial delay, but for the enumeration of the MUSs only incremental polynomial time, since there might be long runs of ``silent paths'' (the negative case, $P >_L P'$, of the test in \textsc{output}, Line \algref{alg:output}{line:ifplspp}):
\begin{theorem}\label{thm:runtimeenumwith}
  \cref{alg:enum}, enumerating $P \in \mf U_x(\idg(F))$ in $L$-pathlex order as delivered to \textsc{output}$(P)$ in Line \algref{alg:DFS}{line:calloutput}, uses time $O(n(F) \ell(F))$ (quadratic time) for finding the first element and for finding each subsequent one, while finally determining that there are no further elements takes time $O(\ell(F))$.
  Space usage is $O(\ell(F))$ (note that paths are not stored).

  For the resulting enumeration of $F' \in \mus_{\set{x}}(F)$ in $L$-pathlex order, each $F'$ has one or two corresponding paths $P$ whose output is $F'$, where only the first path is printed, and the second path is ignored.
  Thus given that we have found $m \in \NNZ$ many MUSs $F'$ so far, finding the next $F'$ (or determining there is none) can be done in time $O((m+1) \cdot n(F) \ell(F))$.
\end{theorem}
\begin{proof}
First we consider the enumeration of the paths $P \in \mf U_x(\idg(F))$.
By Lemma \ref{lem:invariantDFS} we have that each (recursive) call of \textsc{DFS} results in at least one path.
Since the internal runtime of \textsc{DFS} is $O(\ell(F))$, and the maximal depth of a path is $2 n(F)$, we obtain the runtime $O(n(F) \ell(F))$ for finding the next path, while the final phase, determining that there is no further path, is either just backtracking until the root, or just the initial test in Line \algref{alg:enum}{line:initialcheck}, which just takes linear time.
Space usage is that of ordinary DFS (when not copying the path $P$, but implementing it via a global stack).

The worst delay of output of the next MUS is that for all $m$ previous outputs we encounter their (silent) siblings first, and thus we have at most $m+1$ invocations of the search for the next path.
\end{proof}

\begin{question}\label{que:polydelay}
  Can enumeration of $\mus_{\set{x}}(F)$ be done with polynomial delay?
  Once we know one of the underlying paths $P$, then we also know the sibling $P'$, but that comes ``after the fact'', and seems hard to exploit for such printing which avoids long runs of ``silent paths''.
\end{question}

\begin{question}\label{que:usingotherpathenums}
  We use an adaptation of the simplest backtracking algorithm, called \texttt{T-DFS} in \cite{PengEtal2019PathEnumeration}; more sophisticated algorithms don't seem to change the asymptotic bounds, but could be important for practical applications.
\end{question}

Finally, by running through all unit-clauses, we can enumerate all MUSs containing some unit-clause:
\begin{algorithm}[H]
  \caption{Enumeration of $\bc_{\set{x} \in F} \mus_{\set{x}}(F)$ in $L$-pathlex order}
  \label{alg:enumall}
  \begin{algorithmic}[1]
    \Procedure{enum}{$F \in \Bclss, L \in \linord(\lit(F))$}
      \For{$\set{x} \in F$, in order $L$}
        \label{line:enumxL}
        \State $\Call{enum}{F, \set{x}, L}$
        \State $F \gets F \sm \set{\set{x}}$
        \label{line:elimx}
      \EndFor
    \EndProcedure
  \end{algorithmic}
\end{algorithm}

\begin{theorem}\label{thm:correctnessenumwithoutunit}
  \cref{alg:enumall} enumerates $\bc_{\set{x} \in F} \mus_{\set{x}}(F)$ in $L$-pathlex order.
\end{theorem}
\begin{proof}
For different unit-clauses $\set{x}, \set{y} \in F$ we have $\mus_{\set{x},\set{y}}(F) = \mus_{\set{y},\set{x}}(F) \sse \mus_{\set{x}}(F) \cap \mus_{\set{y}}(F)$, and so we need to prevent to output MUSs using $\set{x}$ and $\set{y}$ twice.
By \cref{def:pathlex}, Item \ref{def:pathlex2}, we have to choose the smaller of $x, y$ w.r.t.\ $L$, and that happens in Line \algref{alg:enumall}{line:enumxL}, together with the elimination of it in Line \algref{alg:enumall}{line:elimx}.
Now the assertion follows with \cref{thm:correctnessenumwithunit}.
\end{proof}

\begin{theorem}\label{thm:runtimeenumwithout}
  \cref{alg:enumall}, enumerating $P \in \bc_{\set{x} \in F} \mf U_x(\idg(F))$ in $L$-pathlex order as delivered to \textsc{output}$(P)$ in Line \algref{alg:DFS}{line:calloutput}, uses time $O(n(F) \ell(F))$ for finding the first element and for finding each subsequent one, while finally determining that there are no further elements takes time $O(u(F) \ell(F))$.
  Space usage is $O(\ell(F))$.

  For the resulting enumeration of $F' \in \bc_{\set{x} \in F} \mus_{\set{x}}(F)$ in $L$-pathlex order, given that we have found $m \in \NNZ$ many MUSs $F'$ so far, finding the next $F'$ (or determining there is none) can be done in time $O((m+1) \cdot n(F) \ell(F))$.
\end{theorem}
\begin{proof}
The assertions follow with \cref{thm:runtimeenumwith}, taking into account that for a given unit-clause $\set{x}$ unsatisfiability takes only linear time.
\end{proof}

\section{Summary and Outlook}
\label{sec:outlook}

We started by showing 2-MU decision in linear time (\cref{sec:dec2mu}), based on singular DP-reduction (sDP), with Question \ref{que:complsDP} asking what the general complexity of sDP is.
We then started considering MUSs of 2-CNFs.
Question \ref{que:2MUlintime} asked whether finding some MUS can be done in linear time.
Then in Section \ref{sec:def1MUSNPcompl} we turned to hardness results, showing that deciding the existence of Family-III or Family-IV MUSs is NP-complete.
Family-I and Family-II MUSs (those with two resp.\ one unit-clause) on the other hand can be found efficiently, as shown in Section \ref{Sec:def1MUSsimple}.
Question \ref{que:oneunit} considers the problem of finding shortest Family-II-MUSs (which can be done efficiently for Family-I-MUSs).
Finally we showed that Families I and II together can be enumerated in incremental polynomial time; Question \ref{que:polydelay} asks whether polynomial delay is possible, while Question \ref{que:usingotherpathenums} discusses quickly more practical algorithms.

An interesting open research direction is to establish a more complete \textbf{complexity map} for finding MUSs of 2-CNFs.
The same question about a more complete map arises concerning further efficient enumerations of MUSs.
\begin{conjecture}\label{con:regularpaths}
  Regular paths in digraphs with given skew-symmetry can be enumerated with polynomial delay.
\end{conjecture}
If this conjecture were true, then we could enumerate with polynomial delay all MUSs containing two given unit-clauses.

\newcommand{\noopsort}[1]{}

\begin{appendix}
  
\section{Examples for translations of cycle problems}
\label{sec:appendixexamplescycles}

A simple example of an st-digraph $G$ with a special cycle (recall Definition \ref{def:CDPP}), and its translation $\tFC(G)$ (Definition \ref{def:tFC}) plus the variation $\tFC'(G)$ from the proof of Theorem \ref{thm:gendef1NPCvariant}; the arcs corresponding to the given clauses are drawn with full lines, while their contrapositions are drawn with dashed lines:
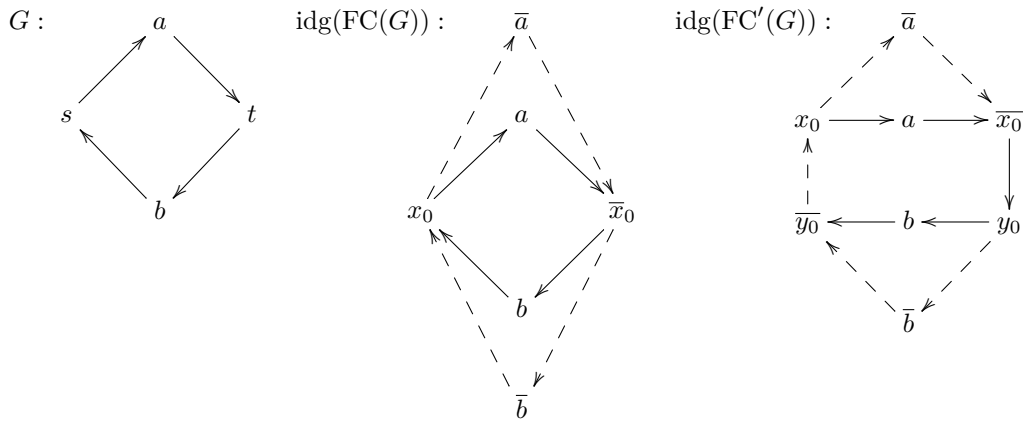
\begin{gather*}
  G:
  \xymatrix {
    & a \ar[dr] & \\
    s \ar[ur] & & t \ar[dl] \\
    & b \ar[ul]  &
  }
  \quad
  \idg(\tFC(G)) :
  \hspace*{-2em}\xymatrix {
    & \ol a\ar@{-->}[ddr] & & \\
    & a \ar[dr] & & \\
    x_0 \ar[ur] \ar@{-->}[uur] & & \ol x_0 \ar[dl] \ar@{-->}[ddl] \\
    & b \ar[ul] & & \\
    & \ol b \ar@{-->}[uul] & &
  }
  \hspace*{-2em}\idg(\tFC'(G)):
  \hspace*{-2em}\xymatrix {
    & \ol a \ar@{-->}[dr] & \\
    x_0 \ar@{-->}[ur] \ar[r] & a \ar[r] & \ol{x_0} \ar[d] \\
    \ol{y_0} \ar@{-->}[u] & b \ar[l] & y_0 \ar[l] \ar@{-->}[dl] \\
    & \ol b \ar@{-->}[ul] &
  }\\
  \tFC(G): x_0 \ra a,\; a \ra \ol{x_0},\; \ol{x_0} \ra b,\; b \ra x_0;\; n = 3, c = 4\\
  \tFC'(G): x_0 \ra a,\; a \ra \ol x_0,\; \boxed{\ol{x_0} \ra y_0},\; y_0 \ra b,\; b \ra \ol{y_0};\; n = 4, c = 5.
\end{gather*}

\section{MUS enumeration algorithm worked example}
\label{sec:appendixenum}

We demonstrate \cref{alg:enum} on the union of three simple MUSs:
\begin{enumerate}
\item $\Hp_{2} = \set{\set{x_1}, \set{\ol{x_1}, x_2}, \set{\ol{x_2}}}$
\item $\Hl_{2,1} = \set{\set{x_1}, \set{\ol{x_1}, x_2},  \set{\ol{x_1}, \ol{x_2}}}$
\item $\Hl_{3,2} = \set{\set{x_1}, \set{\ol{x_1}, x_2}, \set{\ol{x_2}, x_3}, \set{\ol{x_2}, \ol{x_3}}}$.
\end{enumerate}

Their implication digraphs are
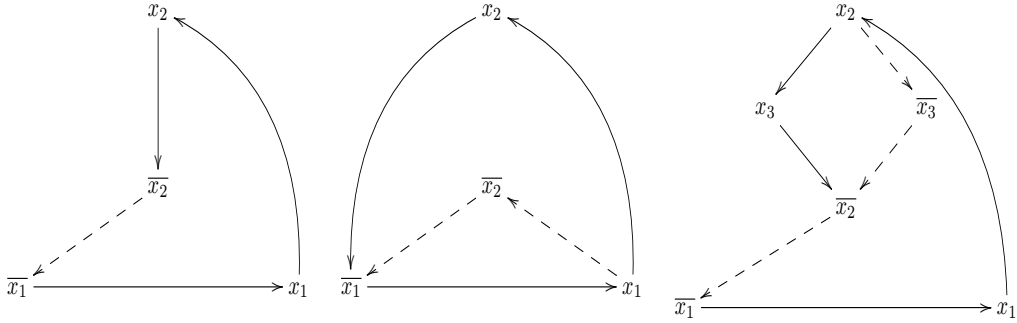
\begin{displaymath}
  \scalebox{0.75}[1]{$
    \xymatrix{
      & & x_2 \ar[dd] & & \\
      & & & & \\
      & & \ol{x_2} \ar@{-->}[dll] & & \\
      \ol{x_1} \ar[rrrr] & & & & x_1 \ar@/_2pc/[uuull]
    } \quad
    \xymatrix{
      & & x_2 \ar@/_2pc/[dddll] & & \\
      & & & & \\
      & & \ol{x_2} \ar@{-->}[dll] & & \\
      \ol{x_1} \ar[rrrr] & & & & x_1 \ar@{-->}[ull] \ar@/_2pc/[uuull]
    } \quad
    \xymatrix{
      & & x_2 \ar@{-->}[dr] \ar[dl] & & \\
      & x_3 \ar[dr] & & \ol{x_3} \ar@{-->}[dl] & \\
      & & \ol{x_2} \ar@{-->}[dll] & & \\
      \ol{x_1} \ar[rrrr] & & & & x_1 \ar@/_2pc/[uuull]
    }
    $}
\end{displaymath}

They belong to families I, IIa, IIb respectively.
The implication digraph of their union is the union of these digraphs:

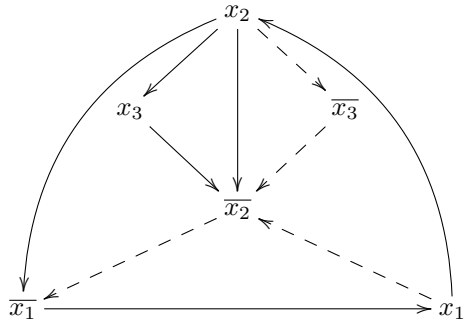
\begin{displaymath}
  \xymatrix{
    & & x_2 \ar@{-->}[dr] \ar[dl] \ar[dd] \ar@/_2pc/[dddll] & & \\
    & x_3 \ar[dr] & & \ol{x_3} \ar@{-->}[dl] & \\
    & & \ol{x_2} \ar@{-->}[dll] & & \\
    \ol{x_1} \ar[rrrr] & & & & x_1 \ar@{-->}[ull] \ar@/_2pc/[uuull]
  }
\end{displaymath}

Running \cref{alg:enum} for $x=x_1$ with $L : (x_1 < \ol{x_1} < x_2 < \ol{x_2} < x_3 < \ol{x_3})$ yields the trace shown in \cref{tab:mus-enum-trace}.
The table lists only the non-trivial events of \cref{alg:DFS,alg:output}.
Steps 5, 8, and 12 produce the three distinct MUSs, while Steps 16 and 20 are lexicographically larger sibling paths for which the guard at line \algref{alg:output}{line:ifplspp} suppresses output.
Columns $y$ and $R$ record, respectively, the current DFS literal (or, at clashes and outputs, the last literal on $P$) and the value computed at line \algref{alg:DFS}{line:computeR}.

\begin{table}
  \centering
  \scriptsize
  \setlength{\tabcolsep}{3pt}
  \caption{Trace of $\enum(F,\set{x_1},L)$ for $F = U^{2}_{2} \cup U^{1}_{2,1} \cup U^{1}_{3,2}$. $P$ records the literal sequence.}
  \label{tab:mus-enum-trace}

  \begin{tabular}{l l}
    \hline
    Event & Algorithmic reference \\
    \hline
    init call & \cref{alg:enum}, line~\ref{line:initialDFS} invoking \cref{alg:DFS}\\
    \textsc{DFS}$(P,y)$ & \cref{alg:DFS}, lines~\ref{line:computeR}--\ref{line:recursiveDFS}: extend $P$, compute $R$, recurse on $y$ \\
    clash & \cref{alg:DFS}, line \ref{line:calloutput}: $V(P)$ contains complementary literals \\
    output & \cref{alg:output}, Family I/IIa/IIb branch (including the guard at line \algref{alg:output}{line:ifplspp}) \\
    output (silent) & \cref{alg:output}, guard line~\ref{line:ifplspp} fails for the lexicographic sibling \\
    \hline
  \end{tabular}
  \medskip

  \begin{tabular}{r l l l l l}
    \hline
    Step & Event & $y$ & $R$ & $P$ & Notes \\
    \hline
    1 & init call & $x_1$ & $(x_2, \ol{x_2})$ & $(x_1)$ & \\
    2 & DFS($P,y$) & $x_2$ & $(\ol{x_1}, \ol{x_2}, x_3, \ol{x_3})$ & $(x_1,x_2)$ & \\
    3 & DFS($P,y$) & $\ol{x_1}$ & -- & $(x_1,x_2,\ol{x_1})$ & \\
    4 & clash & $\ol{x_1}$ & -- & $(x_1,x_2,\ol{x_1})$ & \\
    5 & output & $\ol{x_1}$ & -- & $(x_1,x_2,\ol{x_1})$ & Family IIa\\
    6 & DFS($P,y$) & $\ol{x_2}$ & -- & $(x_1,x_2,\ol{x_2})$ & \\
    7 & clash & $\ol{x_2}$ & -- & $(x_1,x_2,\ol{x_2})$ & \\
    8 & output & $\ol{x_2}$ & -- & $(x_1,x_2,\ol{x_2})$ & Family I\\
    9 & DFS($P,y$) & $x_3$ & $(\ol{x_2})$ & $(x_1,x_2,x_3)$ & \\
    10 & DFS($P,y$) & $\ol{x_2}$ & -- & $(x_1,x_2,x_3,\ol{x_2})$ & \\
    11 & clash & $\ol{x_2}$ & -- & $(x_1,x_2,x_3,\ol{x_2})$ & \\
    12 & output & $\ol{x_2}$ & -- & $(x_1,x_2,x_3,\ol{x_2})$ & Family IIb\\
    13 & DFS($P,y$) & $\ol{x_3}$ & $(\ol{x_2})$ & $(x_1,x_2,\ol{x_3})$ & \\
    14 & DFS($P,y$) & $\ol{x_2}$ & -- & $(x_1,x_2,\ol{x_3},\ol{x_2})$ & \\
    15 & clash & $\ol{x_2}$ & -- & $(x_1,x_2,\ol{x_3},\ol{x_2})$ & \\
    16 & output (silent) & $\ol{x_2}$ & -- & $(x_1,x_2,\ol{x_3},\ol{x_2})$ & Sibling of Step~12; guard suppresses print \\
    17 & DFS($P,y$) & $\ol{x_2}$ & $(\ol{x_1})$ & $(x_1,\ol{x_2})$ & \\
    18 & DFS($P,y$) & $\ol{x_1}$ & -- & $(x_1,\ol{x_2},\ol{x_1})$ & \\
    19 & clash & $\ol{x_1}$ & -- & $(x_1,\ol{x_2},\ol{x_1})$ & \\
    20 & output (silent) & $\ol{x_1}$ & -- & $(x_1,\ol{x_2},\ol{x_1})$ & Sibling of Step~5; guard suppresses print \\
    \hline
  \end{tabular}

  \begin{displaymath}
    \xymatrix{
      & & x_2 \ar@{-->}[dr] \ar[dl] \ar[dd] \ar@/_2pc/[dddll] & & \\
      & x_3 \ar[dr] & & \ol{x_3} \ar@{-->}[dl] & \\
      & & \ol{x_2} \ar@{-->}[dll] & & \\
      \ol{x_1} \ar[rrrr] & & & & x_1 \ar@{-->}[ull] \ar@/_2pc/[uuull]
    }
  \end{displaymath}
\end{table}

\end{appendix}

\end{document}